\documentclass[a4paper,10pt]{article}
\usepackage{amsmath,amssymb,amsfonts,amsthm}
\usepackage[ascii]{inputenc}
\usepackage[margin=4.1cm]{geometry}
\usepackage{bbm}
\usepackage{algorithm}
\usepackage{algorithmic}
\usepackage{graphicx}
\usepackage[pdftex,bookmarks=false,colorlinks=true,linkcolor=blue,
citecolor=blue,filecolor=black,urlcolor=blue]{hyperref}

\providecommand{\ud}{\,\mathrm{d}}
\providecommand{\ii}{\mathbbm{i}}
\providecommand{\vect}[1]{{\boldsymbol{#1}}}
\providecommand{\ket}[1]{\left\lvert#1\right\rangle}
\providecommand{\conj}[1]{\overline{#1}}
\providecommand{\round}[1]{\left[#1\right]}

\newtheorem{definition}{Definition}

\newtheorem{proposition}[definition]{Proposition}
\newtheorem{lemma}[definition]{Lemma}

\begin{document}

\title{\vspace{-2cm}Efficient algorithm for many-electron angular momentum and spin diagonalization \\ on atomic subshells}

\author{Christian B. Mendl\footnote{Zentrum Mathematik, Technische Universit\"at M\"unchen, Boltzmannstra{\ss}e 3, 85748 Garching bei M\"unchen, Germany; \href{mailto:mendl@ma.tum.de}{mendl@ma.tum.de}}}

\date{June 2, 2015}

\maketitle

\begin{abstract}
\footnotesize
We devise an efficient algorithm for the symbolic calculation of irreducible angular momentum and spin (LS) eigenspaces within the $n$-fold antisymmetrized tensor product $\wedge^n V_u$, where $n$ is the number of electrons and $u = \mathrm{s}, \mathrm{p}, \mathrm{d},\dots$ denotes the atomic subshell. This is an essential step for dimension reduction in configuration-interaction (CI) methods applied to atomic many-electron quantum systems. The algorithm relies on the observation that each $L_z$ eigenstate with maximal eigenvalue is also an $\vect{L}^2$ eigenstate (equivalently for $S_z$ and $\vect{S}^2$), as well as the traversal of LS eigenstates using the lowering operators $L_-$ and $S_-$. Iterative application to the remaining states in $\wedge^n V_u$ leads to an \emph{implicit} simultaneous diagonalization. A detailed complexity analysis for fixed $n$ and increasing subshell number $u$ yields run time $\mathcal{O}(u^{3n-2})$. A symbolic computer algebra implementation is available online.

\end{abstract}

\section{Introduction}
\label{sec:Introduction}

Since the inception of quantum mechanics, it is well-known that the (non-relativistic, Born-Oppenheimer) Hamiltonian governing many-electron atoms leaves the simultaneous eigenspaces of the angular momentum, spin and parity (LS) operators
\begin{equation}
\label{eq:LSOperators}
\vect{L}^2, \ L_z, \ \vect{S}^2, \ S_z, \ \hat{R}
\end{equation}
invariant. From a practical perspective, the restriction to symmetry subspaces can significantly reduce computational costs (see, e.g., Ref.~\cite{FroeseFischer1997, FragaKarwowskiSaxena1976, Taylor1986, ChaichianHagedornBook1997}). In particular, such a restriction is an essential ingredient for configuration interaction~(CI) approximation methods in Ref.~\cite{NuclearChargeLimit2009,AsymptoticsCI2009,Chromium2010}. However, simultaneous diagonalization of the operators~\eqref{eq:LSOperators} on the full CI space is encumbered by the inherent ``curse of dimensionality'', which renders ``naive'' $\mathcal{O}(\dim^3)$ approaches infeasible. The present paper outlines an efficient algorithm for computing the symbolic eigenspaces by making use of representation theory and the algebraic properties of the LS operators.

In \eqref{eq:LSOperators}, the total angular momentum operator is defined as $\vect{L} = \sum_{j=1}^n \vect{L}(j)$ with $n$ the number of electrons and
\begin{equation}
\vect{L}(j) = \tfrac{1}{\mathbbm{i}}\, \vect{x}_j \times \vect{\nabla}_j
\end{equation}
the angular momentum operator acting on electron $j$. (We choose units such that $\hbar = 1$.) $L_z$ is the third component of $\vect{L}$. In spherical polar coordinates, $L_z(j) = \tfrac{1}{\mathbbm{i}} \partial/\partial \varphi_j$. Analogously for spin, $\vect{S} = \sum_{j=1}^n \vect{S}(j)$ with $S_{\alpha}(j)$ for $\alpha = x, y, z$ the usual Pauli matrices
\begin{equation}
\sigma_x = \frac{1}{2} \begin{pmatrix}0 & 1 \\ 1 & 0\end{pmatrix}, \quad \sigma_y = \frac{1}{2} \begin{pmatrix}0 & -\ii \\ \ii & 0\end{pmatrix}, \quad \sigma_z = \frac{1}{2} \begin{pmatrix}1 & 0 \\ 0 & -1\end{pmatrix}
\end{equation}
acting on electron $j$. The components of the angular and spin operators obey the well-known commutator relations $[L_{\alpha}, L_{\beta}] = \ii L_{\gamma}$ and $[S_{\alpha}, S_{\beta}] = \ii S_{\gamma}$ with $\alpha, \beta, \gamma$ cyclic permutations of $x, y, z$. The \emph{ladder operators} are given by $L_{\pm} = L_x \pm \ii L_y$ and $S_{\pm} = S_x \pm \ii S_y$. They have the property that for any angular momentum eigenfunction $\psi^{m_\ell}$ with eigenvalue $m_\ell$, $L_{\pm} \psi^{m_\ell}$ is zero or an eigenfunction with eigenvalue $m_\ell \pm 1$, and correspondingly for spin. The \emph{parity operator} acts on wavefunctions as $\hat{R}\,\psi(\vect{x}_1, s_1, \dots, \vect{x}_n, s_n) = \psi(-\vect{x}_1, s_1, \dots, -\vect{x}_n, s_n)$, where $\vect{x}_j \in \mathbb{R}^3$ and $s_j \in \{-\frac{1}{2}, \frac{1}{2}\}$ are the position and spin coordinate of electron $j$.

The simultaneous diagonalization of the LS operators is greatly simplified by representation theory using Clebsch-Gordan coefficients. Specifically, the required computational cost is reduced to the calculation of irreducible LS representation spaces (i.e., diagonalizing the operators \eqref{eq:LSOperators}) on the $n$-fold antisymmetrized tensor product $\wedge^n V_u$ (compare with Ref.~\cite[proposition~2]{Chromium2010}). Here, $V_u$ denotes an angular momentum subshell, $u = \mathrm{s}, \mathrm{p}, \mathrm{d}, \mathrm{f},\dots$ in chemist's notation. An explicit realization of $V_u$ is
\begin{equation}
V_u = \mathrm{span}\left\{Y_{u,m}\!\uparrow,Y_{u,m}\!\downarrow\right\}_{m=u, u-1, \dots, -u}
\end{equation}
with the spherical harmonics $Y_{u,m}$:
\begin{equation*}
\begin{split}
Y_{\mathrm{s},0} &= \tfrac{1}{\sqrt{4 \pi}}, \\
Y_{\mathrm{p},1} = -\tfrac{1}{2} \sqrt{\tfrac{3}{2 \pi}} \sin(\theta) \mathrm{e}^{\ii \varphi}, \quad Y_{\mathrm{p},0} &= \tfrac{1}{2} \sqrt{\tfrac{3}{\pi}} \cos(\theta), \quad Y_{\mathrm{p},-1} = \tfrac{1}{2} \sqrt{\tfrac{3}{2 \pi}} \sin(\theta) \mathrm{e}^{-\ii \varphi }\\
&\dots
\end{split}
\end{equation*}
We identify the subshell label $u$ with the corresponding quantum number, i.e., $\mathrm{s},\mathrm{p},\mathrm{d},\mathrm{f},\dots{}\leftrightarrow 0,1,2,3,\dots$ In particular, $\dim(V_u) = 2\,(2\,u + 1)$. Note that $Y_{u,m}\!\uparrow$, $Y_{u,m}\!\downarrow$ are simultaneous single-particle $L_z$-$S_z$ eigenstates. They serve as underlying ordered orbitals, which we denote abstractly as
\begin{align*}
&\left(\mathrm{s}, \conj{\mathrm{s}}\right) &\text{for } & V_{\mathrm{s}},\\
&\left(\mathrm{p}_1, \conj{\mathrm{p}_1}, \mathrm{p}_0, \conj{\mathrm{p}_0}, \mathrm{p}_{\text{-}1}, \conj{\mathrm{p}_{\text{-}1}}\right) &\text{for } & V_{\mathrm{p}},\\
&\left(\mathrm{d}_2, \conj{\mathrm{d}_2}, \mathrm{d}_1, \conj{\mathrm{d}_1},\dots, \mathrm{d}_{\text{-}2}, \conj{\mathrm{d}_{\text{-}2}}\right) &\text{for } & V_{\mathrm{d}},\\
&\left(\mathrm{f}_3, \conj{\mathrm{f}_3}, \mathrm{f}_2, \conj{\mathrm{f}_2},\dots, \mathrm{f}_{\text{-}3}, \conj{\mathrm{f}_{\text{-}3}}\right) &\text{for } & V_{\mathrm{f}},\\
&\qquad\dots
\end{align*}
The highest $L_z$ quantum number appears first, and $\conj{\,\cdot\,}$ equals spin down $\downarrow$, following the convention in Ref.~\cite{NuclearChargeLimit2009}. The elements of $\wedge^n V_u$ are then linear combinations of Slater determinants built from these orbitals, for example $\frac{1}{\sqrt{2}} \lvert \mathrm{d}_2 \conj{\mathrm{d}_1} \mathrm{d}_{\text{-}1}\rangle - \frac{\ii}{\sqrt{2}} \lvert\mathrm{d}_1 \mathrm{d}_0 \conj{\mathrm{d}_0}\rangle \in \wedge^3 V_{\mathrm{d}}$.

The simultaneous diagonalization may now be formalized as follows. For a given $n \in \{1,2,\dots,\dim(V_u)\}$, we need to decompose the $n$-particle space $\wedge^n V_u$ into irreducible LS representation spaces $V_{u,n,i}$,
\begin{equation}
\label{eq:IrredDecompose}
\wedge^n V_u = \bigoplus_i V_{u,n,i}
\end{equation}
such that
\begin{equation}
\label{eq:IrredSubspace}
\begin{split}
\vect{L}^2\,\varphi &= \ell_i (\ell_i+1)\,\varphi, \quad L_\pm\,\varphi \in V_{u,n,i},\\
\vect{S}^2\,\varphi &= s_i (s_i+1)\,\varphi, \quad S_\pm\,\varphi \in V_{u,n,i} \quad \text{for all}\ \varphi \in V_{u,n,i},\\
\dim(V_{u,n,i}) &= (2\ell_i+1) (2s_i+1).
\end{split}
\end{equation}

\begin{figure}[!ht]
\centering
\includegraphics[width=0.6\textwidth]{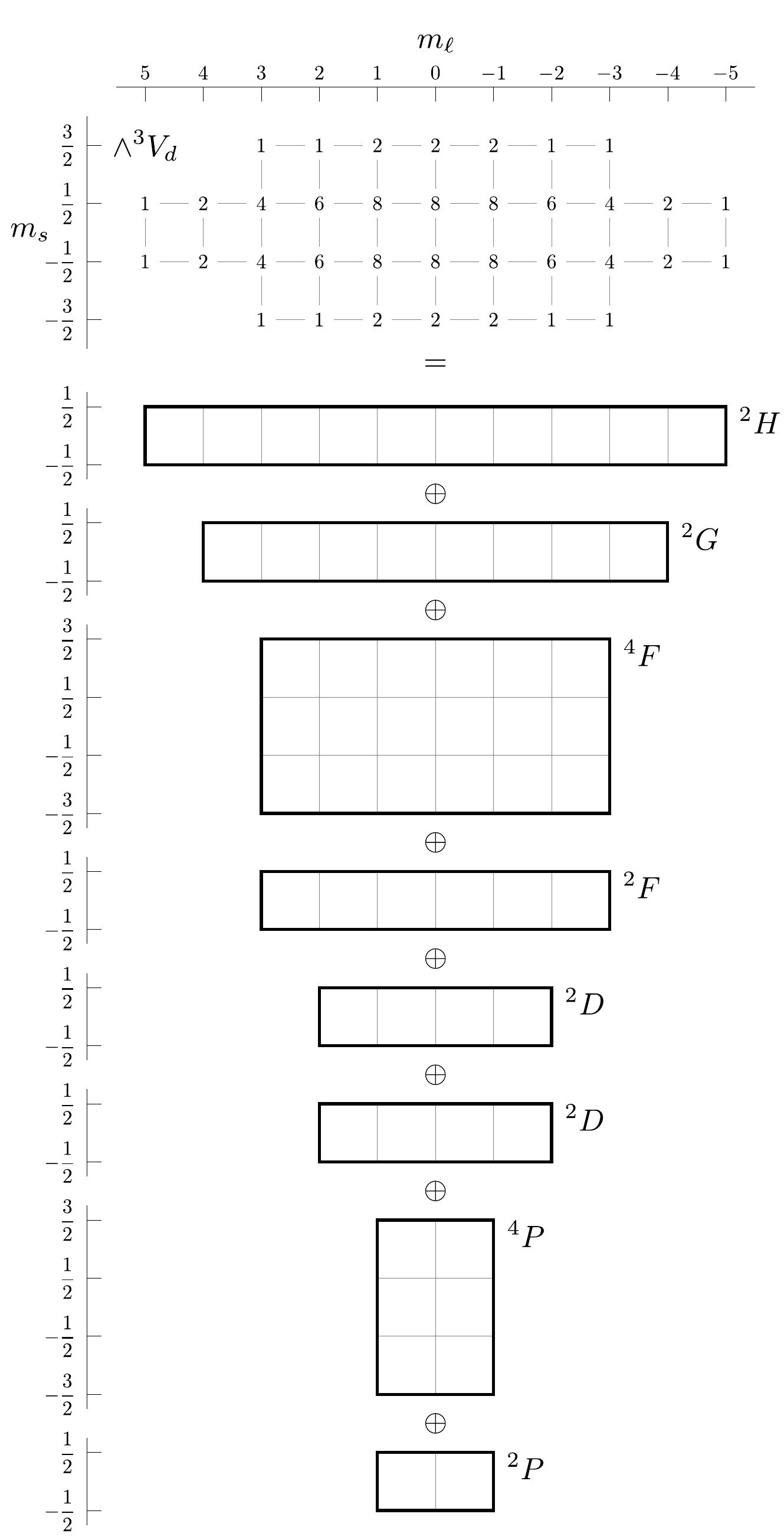}
\caption{Algebraic decomposition of $\wedge^3 V_{\mathrm{d}}$ into irreducible LS representation spaces $V_{\mathrm{d},3,i}$ (see equation~\eqref{eq:IrredDecompose}). Each of these spaces corresponds to a rectangle, matching the $m_\ell$ and $m_s$ quantum numbers running from $\ell_i,\dots,-\ell_i$ and $s_i,\dots,-s_i$, respectively. The $(\ell_i,s_i)$ quantum numbers are displayed in common chemist's notation as $^{2s+1}\ell$. Superimposing all rectangles yields the multiplicities of the $L_z$-$S_z$ eigenvalues in the table at the top.}
\label{fig:quantLS}
\end{figure}

The proposed algorithm (see section~\ref{sec:Algorithm}) performs the LS diagonalization implicitly, relies on the sparse matrix structure of the lowering operators $L_-$, $S_-$, and makes use of the algebraic structure of $\wedge^n V_u$ as illustrated in figure~\ref{fig:quantLS}. We present explicit tables containing decompositions of selected $\wedge^n V_u$ in section~\ref{sec:Results}. Given $u$, the number of electrons maximizing $\dim(\wedge^n V_u) = \binom{\dim(V_u)}{n}$ equals $n = 2\,u+1$ since $\dim(V_u) = 2\,(2\,u + 1)$. Due to this exponential growth in $u$, solving Eq.~\eqref{eq:IrredDecompose} for all possible $n$ restricts $u$ to the $\mathrm{s}$, $\mathrm{p}$ and $\mathrm{d}$ subshells at present, and $u = \mathrm{f}$ for all $n = 1, \dots, 14$ might still be attainable. On the other hand, keeping $n$ fixed means that $\dim(\wedge^n V_u) = \mathcal{O}(u^n)$ asymptotically in $u$. For a given $n$ the algorithm has run time
\begin{equation}
R_n(u) = \mathcal{O}\big(u^{3n-2}\big),
\end{equation}
as derived in section~\ref{sec:RunTime}. In particular, for $n = 2$, this equals $\mathcal{O}(\dim(\wedge^n V_u)^2)$ (instead of $\mathcal{O}(\dim(\wedge^n V_u)^3)$ for the usual diagonalization of a dense matrix).

As an alternative scenario, consider the case that we are only interested in representation spaces $V_{u,n,i}$ with $\ell_i$ and $s_i$ equal (or close to) zero. As our analysis will show, this opens up the possibility of explicitly diagonalizing~\eqref{eq:LSOperators} restricted to the ``central'' simultaneous $L_z$-$S_z$ eigenspace with eigenvalues $(0,0)$ for $n$ even and $(0,\frac{1}{2})$ for $n$ odd, respectively. Due to symmetry, this eigenspace also has the highest dimension (denoted $d_{u,n}$) among all simultaneous $L_z$-$S_z$ eigenspaces within $\wedge^n V_u$. In section~\ref{sec:d_un}, we derive the asymptotic result
\begin{equation}
\label{eq:d_un}
d_{u,n} \cong \sqrt3\ \frac{\dim\left(\wedge^n V_u\right)}{\pi\,n\,u} = \mathcal{O}\big(u^{n-1}\big) \quad \text{as } u \to \infty, \text{ for fixed } n.
\end{equation}
Thus, diagonalization restricted to this central eigenspace still requires $\mathcal{O}(d_{u,n}^3) = \mathcal{O}(u^{3n-3})$ operations.

\section{Algorithm}
\label{sec:Algorithm}

The reasoning and basic ingredients of our algorithm are as follows:
\begin{enumerate}
\item Observe that the canonical Slater determinant basis vectors of $\wedge^n V_u$ are precisely the eigenvectors of both $L_z$ and $S_z$ acting on $\wedge^n V_u$. For example, $L_z\,\ket{\mathrm{d}_2 \conj{\mathrm{d}_1} \mathrm{d}_{\text{-}1}} = (2 + 1 - 1)\ket{\mathrm{d}_2 \conj{\mathrm{d}_1} \mathrm{d}_{\text{-}1}}$ and $S_z\,\ket{\mathrm{d}_2 \conj{\mathrm{d}_1} \mathrm{d}_{\text{-}1}} = (\frac{1}{2} - \frac{1}{2} + \frac{1}{2})\ket{\mathrm{d}_2 \conj{\mathrm{d}_1} \mathrm{d}_{\text{-}1}}$. In particular, all simultaneous $L_z$-$S_z$ eigen\emph{values} can easily be enumerated, including multiplicities.

\item Let $\ell_{\max}$ be the largest $L_z$ eigenvalue on $\wedge^n V_u$ and $W_{L_z,\max}$ the corresponding eigenspace, as well as $\psi \in W_{L_z,\max} \setminus\{0\}$. Then $\psi$ must also be an $\vect{L}^2$ eigenvector with eigenvalue $\ell_{\max}(\ell_{\max}+1)$. This follows from the identity
\begin{equation}
\vect{L}^2 = L_z (L_z + \boldsymbol{1}) + L_{-} L_{+}
\end{equation}
and the fact that $L_+$ is zero on $W_{L_z,\max}$ since $\ell_{\max}$ is -- by definition -- the largest $L_z$ eigenvalue. The same reasoning applies to $S_z$ and $\vect{S}^2$ restricted to $W_{L_z,\max}$. Thus we may assume that $\psi$ is also a $S_z$-$\vect{S}^2$ eigenvector with eigenvalue $s$ and $s(s+1)$, respectively.

\item Starting from $\psi$, we may span an irreducible LS representation space $V_\psi$ by repeatedly applying the lowering operators $L_-$ and $S_-$. That is, $V_\psi := \mathrm{span}\{\psi, L_- \psi, S_- \psi, L_- S_- \psi, \dots\}$.

\item We obtain all remaining irreducible representation spaces by iteratively applying steps 2 and 3 to the orthogonal complement of $V_\psi$ in $\wedge^n V_u$.
\end{enumerate}

Note that although the underlying Hilbert space is complex, all steps involve real-valued matrix representations of the operators $L_z, S_z, L_\pm, S_\pm$ only. Thus, the whole algorithm can be implemented on the real numbers.

\begin{algorithm}
\caption{Quantum numbers of all irreducible subspaces in $\wedge^n V_u$}
\label{alg:quantLS}
\begin{algorithmic}[1]
\STATE Enumerate the simultaneous eigenvalues of $L_z$ and $S_z$ acting on $\wedge^n V_u$, including multiplicities, and store them in a table denoted $T_z$. For example, figure~\ref{fig:quantLS} shows the multiplicity table for $\wedge^3 V_{\mathrm{d}}$.
\STATE $i \gets 1$
\WHILE{$T_z$ contains non-zero multiplicities}
	\STATE Let $\ell := \ell_{\max}$ be the greatest $L_z$ eigenvalue in $T_z$ with non-zero multiplicity, and let $s$ be a corresponding $S_z$ eigenvalue which is maximal among all tuples $(\ell_{\max},s)$ in $T_z$.
	\STATE Calculate the $m_\ell$ and $m_s$ quantum numbers corresponding to $(\ell,s)$, i.e., the tuples $(m_\ell,m_s)$ for all $m_\ell = \ell,\dots,-\ell$ and $m_s = s,\dots,-s$. Decrement the multiplicity of each $(m_\ell,m_s)$ in $T_z$ by one.
	\STATE $(\ell_i,s_i) \gets (\ell,s)$ (store the current quantum numbers), and increment $i$.
\ENDWHILE
\end{algorithmic}
\end{algorithm}

The $L_z$-$S_z$ quantum numbers (including multiplicities) are sufficient to calculate the $(\ell_i,s_i)$ quantum numbers in Eq.~\eqref{eq:IrredSubspace}, see algorithm~\ref{alg:quantLS}. Since each irreducible LS space contains exactly one vector in the ``central'' simultaneous $L_z$-$S_z$ eigenspace with eigenvalues $(0,0)$ ($n$ even) or $(0,\frac{1}{2})$ ($n$ odd) and multiplicity $d_{u,n}$, there are exactly $d_{u,n}$ irreducible LS spaces.

Algorithm~\ref{alg:IrredLS} actually performs the simultaneous diagonalization. It requires the $(\ell_i,s_i)$ tuples computed by algorithm~\ref{alg:quantLS}.

\begin{algorithm}
\caption{Simultaneous diagonalization of the operators~\eqref{eq:LSOperators} on $\wedge^n V_u$, yielding the decomposition~\eqref{eq:IrredDecompose}}
\label{alg:IrredLS}
\begin{algorithmic}[1]
\REQUIRE Irreducible representation space quantum numbers $(\ell_i,s_i)$ as computed by algorithm~\ref{alg:quantLS}.
\STATE Partition the canonical Slater determinant basis of $\wedge^n V_u$ into simultaneous $L_z$-$S_z$ eigenspaces denoted $W_{m_\ell,m_s}$. That is, $W_{m_\ell,m_s}$ is the eigenspace corresponding to eigenvalues $m_\ell$ and $m_s$, respectively.
\FOR{$i = 1,2,\dots$}
\STATE
\label{algline:LoweringSpan}
Select a (normalized) $\psi_i \in W_{\ell_i,s_i}$ and span the corresponding irreducible representation space $V_{u,n,i}$ in~\eqref{eq:IrredDecompose} by repeatedly applying the lowering operators $L_-$ and $S_-$. That is,
\begin{align*}
&V_{u,n,i} := \mathrm{span}\left\{\psi_i^{m_\ell,m_s}\right\}_{m_\ell=\ell_i,\dots,-\ell_i,m_s=s_i,\dots,-s_i} \quad \text{with}\\
&\psi_i^{\ell_i,s_i} := \psi_i \quad \text{and}\\
&\psi_i^{m_\ell-1,m_s} := c_{\ell_i,m_\ell} L_-\,\psi_i^{m_\ell,m_s},\\
&\psi_i^{m_\ell,m_s-1} := c_{s_i,m_s} S_-\,\psi_i^{m_\ell,m_s}
\end{align*}
and the normalization factors $c_{\ell,m} := (\ell(\ell+1)-m(m-1))^{-1/2}$.
\STATE Remove the vectors spanning $V_{u,n,i}$ from any corresponding $L_z$-$S_z$ eigenspace $W_{\ell_j,s_j}$ with $\ell_j \le \ell_i$ and $s_j \le s_i$. More precisely, update $W_{\ell_j,s_j}$ such that it contains the orthogonal complement of $\psi_i^{\ell_j,s_j}$ in $W_{\ell_j,s_j}$. \label{algline:OrthComplement}
\ENDFOR
\end{algorithmic}
\end{algorithm}

The basis vectors spanning the orthogonal complement in $W_{\ell_j,s_j}$ (line~\ref{algline:OrthComplement}) are not unique. This poses a practical problem for symbolic computer algebra implementations. Namely, orthonormalizing these basis vectors can lead to a blow-up of nested squares, which is particularly unfavorable since subsequently the lowering operators (line~\ref{algline:LoweringSpan}) are applied to these vectors. To circumvent this difficulty, one can instead work with the unique projection matrix $P_j$ acting on the basis vectors initially in $W_{\ell_j,s_j}$. Then, in line~\ref{algline:OrthComplement}, $P_j$ is updated such that it spans precisely the orthogonal complement:
\begin{equation}
P_j \gets P_j - \big\lvert \psi_i^{\ell_j,s_j}\big\rangle \big\langle \psi_i^{\ell_j,s_j} \big\rvert.
\end{equation}
At the beginning of the algorithm, each $P_j$ starts as identity matrix (on $W_{\ell_j,s_j}$), and ends as zero matrix.

\section{Example decompositions}
\label{sec:Results}

Explicit decompositions of $\wedge^n V_{\mathrm{f}}$ for $n = 1,2,3$ are shown in table~\ref{tab:irredLSf}. We have omitted $\wedge^n V_u$, $u = \mathrm{s}, \mathrm{p}, \mathrm{d}$ since these are already published in~\cite{Chromium2010}. The complete tables are available online, including a \textsf{Mathematica} implementation of the algorithm \cite{irredLSGithub} which makes use of the \textsf{FermiFab} toolbox \cite{FermiFabPaper2011, FermiFabSoftware}. For conciseness, only states with maximal $L_z$ and $S_z$ quantum numbers are displayed; applying the lowering operators $L_{-}$ and $S_{-}$ yields the remaining wavefunctions. Note that in general, symmetry levels can appear more than once within a many-particle subshell, e.g., $^2\mathrm{G}^{\mathrm{o}}$ in $\wedge^3 V_\mathrm{f}$. Thus, the tables are only unique up to (orthogonal) base changes of the states within the same symmetry level. The run time on a commodity laptop computer to calculate the symbolic eigenspaces is approximately 16 seconds for $u = \mathrm{f}$ and $n = 3$, and 550 seconds for $u = \mathrm{f}$ and $n = 4$.

\begin{table}[!htp]
\centering
\tiny
\begin{tabular}{|c|c|cc|c|}
\hline
config&sym&$L_z$&$\mathrm{S}_z$&$\Psi$ $\vphantom{\biggl(}$\\
\hline\hline
$\wedge^{1}V_{\mathrm{f}}$&$^2\mathrm{F}^{\mathrm{o}}$&$3$&$\frac{1}{2}$&$\ket{\mathrm{f}_3}$ $\vphantom{\biggl(}$\\
\hline
$\wedge^{2}V_{\mathrm{f}}$&$^1\mathrm{I}$&$6$&$0$&$\ket{\mathrm{f}_3 \conj{\mathrm{f}_3}}$ $\vphantom{\biggl(}$\\
\cline{2-5}
&$^3\mathrm{H}$&$5$&$1$&$\ket{\mathrm{f}_3 \mathrm{f}_2}$ $\vphantom{\biggl(}$\\
\cline{2-5}
&$^1\mathrm{G}$&$4$&$0$&$\frac{1}{\sqrt{11}}\left(-\sqrt{3}\cdot\ket{\mathrm{f}_3 \conj{\mathrm{f}_1}}+\sqrt{3}\cdot\ket{\conj{\mathrm{f}_3} \mathrm{f}_1}+\sqrt{5}\cdot\ket{\mathrm{f}_2 \conj{\mathrm{f}_2}}\right)$ $\vphantom{\biggl(}$\\
\cline{2-5}
&$^3\mathrm{F}$&$3$&$1$&$\frac{1}{\sqrt{3}}\left(-\ket{\mathrm{f}_3 \mathrm{f}_0}+\sqrt{2}\cdot\ket{\mathrm{f}_2 \mathrm{f}_1}\right)$ $\vphantom{\biggl(}$\\
\cline{2-5}
&$^1\mathrm{D}$&$2$&$0$&$\frac{1}{\sqrt{42}}\left(\sqrt{5}\cdot\ket{\mathrm{f}_3 \conj{\mathrm{f}_{\text{-}1}}}-\sqrt{5}\cdot\ket{\conj{\mathrm{f}_3} \mathrm{f}_{\text{-}1}}-\sqrt{10}\cdot\ket{\mathrm{f}_2 \conj{\mathrm{f}_0}}+\sqrt{10}\cdot\ket{\conj{\mathrm{f}_2} \mathrm{f}_0}+2 \sqrt{3}\cdot\ket{\mathrm{f}_1 \conj{\mathrm{f}_1}}\right)$ $\vphantom{\biggl(}$\\
\cline{2-5}
&$^3\mathrm{P}$&$1$&$1$&$\frac{1}{\sqrt{14}}\left(\sqrt{3}\cdot\ket{\mathrm{f}_3 \mathrm{f}_{\text{-}2}}-\sqrt{5}\cdot\ket{\mathrm{f}_2 \mathrm{f}_{\text{-}1}}+\sqrt{6}\cdot\ket{\mathrm{f}_1 \mathrm{f}_0}\right)$ $\vphantom{\biggl(}$\\
\cline{2-5}
&$^1\mathrm{S}$&$0$&$0$&$\frac{1}{\sqrt{7}}\left(-\ket{\mathrm{f}_3 \conj{\mathrm{f}_{\text{-}3}}}+\ket{\conj{\mathrm{f}_3} \mathrm{f}_{\text{-}3}}+\ket{\mathrm{f}_2 \conj{\mathrm{f}_{\text{-}2}}}-\ket{\conj{\mathrm{f}_2} \mathrm{f}_{\text{-}2}}-\ket{\mathrm{f}_1 \conj{\mathrm{f}_{\text{-}1}}}+\ket{\conj{\mathrm{f}_1} \mathrm{f}_{\text{-}1}}+\ket{\mathrm{f}_0 \conj{\mathrm{f}_0}}\right)$ $\vphantom{\biggl(}$\\
\hline
$\wedge^{3}V_{\mathrm{f}}$&$^2\mathrm{K}^{\mathrm{o}}$&$8$&$\frac{1}{2}$&$\ket{\mathrm{f}_3 \conj{\mathrm{f}_3} \mathrm{f}_2}$ $\vphantom{\biggl(}$\\
\cline{2-5}
&$^2\mathrm{J}^{\mathrm{o}}$&$7$&$\frac{1}{2}$&$\frac{1}{2 \sqrt{2}}\left(\sqrt{3}\cdot\ket{\mathrm{f}_3 \conj{\mathrm{f}_3} \mathrm{f}_1}+\sqrt{5}\cdot\ket{\mathrm{f}_3 \mathrm{f}_2 \conj{\mathrm{f}_2}}\right)$ $\vphantom{\biggl(}$\\
\cline{2-5}
&$^4\mathrm{I}^{\mathrm{o}}$&$6$&$\frac{3}{2}$&$\ket{\mathrm{f}_3 \mathrm{f}_2 \mathrm{f}_1}$ $\vphantom{\biggl(}$\\
\cline{2-5}
&$^2\mathrm{I}^{\mathrm{o}}$&$6$&$\frac{1}{2}$&$\frac{1}{\sqrt{21}}\left(3\cdot\ket{\mathrm{f}_3 \conj{\mathrm{f}_3} \mathrm{f}_0}-\sqrt{2}\cdot\ket{\mathrm{f}_3 \mathrm{f}_2 \conj{\mathrm{f}_1}}-\sqrt{2}\cdot\ket{\mathrm{f}_3 \conj{\mathrm{f}_2} \mathrm{f}_1}+2 \sqrt{2}\cdot\ket{\conj{\mathrm{f}_3} \mathrm{f}_2 \mathrm{f}_1}\right)$ $\vphantom{\biggl(}$\\
\cline{2-5}
&$^2\mathrm{H}^{\mathrm{o}}$&$5$&$\frac{1}{2}$&$\frac{1}{\sqrt{6}}\left(\sqrt{2}\cdot\ket{\mathrm{f}_3 \conj{\mathrm{f}_3} \mathrm{f}_{\text{-}1}}-\ket{\mathrm{f}_3 \conj{\mathrm{f}_2} \mathrm{f}_0}+\ket{\conj{\mathrm{f}_3} \mathrm{f}_2 \mathrm{f}_0}+\sqrt{2}\cdot\ket{\mathrm{f}_2 \conj{\mathrm{f}_2} \mathrm{f}_1}\right)$ $\vphantom{\biggl(}$\\
\cline{2-5}
&$^2\mathrm{H}^{\mathrm{o}}$&$5$&$\frac{1}{2}$&$\frac{1}{\sqrt{273}}\left(-\sqrt{5}\cdot\ket{\mathrm{f}_3 \conj{\mathrm{f}_3} \mathrm{f}_{\text{-}1}}-3 \sqrt{10}\cdot\ket{\mathrm{f}_3 \mathrm{f}_2 \conj{\mathrm{f}_0}}+2 \sqrt{10}\cdot\ket{\mathrm{f}_3 \conj{\mathrm{f}_2} \mathrm{f}_0}+6 \sqrt{3}\cdot\ket{\mathrm{f}_3 \mathrm{f}_1 \conj{\mathrm{f}_1}}\right.$ $\vphantom{\biggl(}$\\
&&&&$\left.+\sqrt{10}\cdot\ket{\conj{\mathrm{f}_3} \mathrm{f}_2 \mathrm{f}_0}+2 \sqrt{5}\cdot\ket{\mathrm{f}_2 \conj{\mathrm{f}_2} \mathrm{f}_1}\right)$ $\vphantom{\biggl(}$\\
\cline{2-5}
&$^4\mathrm{G}^{\mathrm{o}}$&$4$&$\frac{3}{2}$&$\frac{1}{\sqrt{11}}\left(-\sqrt{5}\cdot\ket{\mathrm{f}_3 \mathrm{f}_2 \mathrm{f}_{\text{-}1}}+\sqrt{6}\cdot\ket{\mathrm{f}_3 \mathrm{f}_1 \mathrm{f}_0}\right)$ $\vphantom{\biggl(}$\\
\cline{2-5}
&$^2\mathrm{G}^{\mathrm{o}}$&$4$&$\frac{1}{2}$&$\frac{1}{7 \sqrt{5}}\left(5 \sqrt{3}\cdot\ket{\mathrm{f}_3 \conj{\mathrm{f}_3} \mathrm{f}_{\text{-}2}}+\sqrt{5}\cdot\ket{\mathrm{f}_3 \mathrm{f}_2 \conj{\mathrm{f}_{\text{-}1}}}-3 \sqrt{5}\cdot\ket{\mathrm{f}_3 \conj{\mathrm{f}_2} \mathrm{f}_{\text{-}1}}-\sqrt{6}\cdot\ket{\mathrm{f}_3 \mathrm{f}_1 \conj{\mathrm{f}_0}}\right.$ $\vphantom{\biggl(}$\\
&&&&$\left.+\sqrt{6}\cdot\ket{\mathrm{f}_3 \conj{\mathrm{f}_1} \mathrm{f}_0}+2 \sqrt{5}\cdot\ket{\conj{\mathrm{f}_3} \mathrm{f}_2 \mathrm{f}_{\text{-}1}}+2 \sqrt{10}\cdot\ket{\mathrm{f}_2 \conj{\mathrm{f}_2} \mathrm{f}_0}+4 \sqrt{3}\cdot\ket{\mathrm{f}_2 \mathrm{f}_1 \conj{\mathrm{f}_1}}\right)$ $\vphantom{\biggl(}$\\
\cline{2-5}
&$^2\mathrm{G}^{\mathrm{o}}$&$4$&$\frac{1}{2}$&$\frac{1}{7 \sqrt{429}}\left(-18 \sqrt{6}\cdot\ket{\mathrm{f}_3 \conj{\mathrm{f}_3} \mathrm{f}_{\text{-}2}}+16 \sqrt{10}\cdot\ket{\mathrm{f}_3 \mathrm{f}_2 \conj{\mathrm{f}_{\text{-}1}}}+\sqrt{10}\cdot\ket{\mathrm{f}_3 \conj{\mathrm{f}_2} \mathrm{f}_{\text{-}1}}-32 \sqrt{3}\cdot\ket{\mathrm{f}_3 \mathrm{f}_1 \conj{\mathrm{f}_0}}\right.$ $\vphantom{\biggl(}$\\
&&&&$\left.-17 \sqrt{3}\cdot\ket{\mathrm{f}_3 \conj{\mathrm{f}_1} \mathrm{f}_0}-17 \sqrt{10}\cdot\ket{\conj{\mathrm{f}_3} \mathrm{f}_2 \mathrm{f}_{\text{-}1}}+49 \sqrt{3}\cdot\ket{\conj{\mathrm{f}_3} \mathrm{f}_1 \mathrm{f}_0}+15 \sqrt{5}\cdot\ket{\mathrm{f}_2 \conj{\mathrm{f}_2} \mathrm{f}_0}\right.$ $\vphantom{\biggl(}$\\
&&&&$\left.+15 \sqrt{6}\cdot\ket{\mathrm{f}_2 \mathrm{f}_1 \conj{\mathrm{f}_1}}\right)$ $\vphantom{\biggl(}$\\
\cline{2-5}
&$^4\mathrm{F}^{\mathrm{o}}$&$3$&$\frac{3}{2}$&$\frac{1}{2}\left(\ket{\mathrm{f}_3 \mathrm{f}_2 \mathrm{f}_{\text{-}2}}-\ket{\mathrm{f}_3 \mathrm{f}_1 \mathrm{f}_{\text{-}1}}+\sqrt{2}\cdot\ket{\mathrm{f}_2 \mathrm{f}_1 \mathrm{f}_0}\right)$ $\vphantom{\biggl(}$\\
\cline{2-5}
&$^2\mathrm{F}^{\mathrm{o}}$&$3$&$\frac{1}{2}$&$\frac{1}{\sqrt{6}}\left(\ket{\mathrm{f}_3 \conj{\mathrm{f}_3} \mathrm{f}_{\text{-}3}}+\ket{\mathrm{f}_3 \mathrm{f}_2 \conj{\mathrm{f}_{\text{-}2}}}-\ket{\mathrm{f}_3 \conj{\mathrm{f}_2} \mathrm{f}_{\text{-}2}}-\ket{\mathrm{f}_3 \mathrm{f}_1 \conj{\mathrm{f}_{\text{-}1}}}\right.$ $\vphantom{\biggl(}$\\
&&&&$\left.+\ket{\mathrm{f}_3 \conj{\mathrm{f}_1} \mathrm{f}_{\text{-}1}}+\ket{\mathrm{f}_3 \mathrm{f}_0 \conj{\mathrm{f}_0}}\right)$ $\vphantom{\biggl(}$\\
\cline{2-5}
&$^2\mathrm{F}^{\mathrm{o}}$&$3$&$\frac{1}{2}$&$\frac{1}{2 \sqrt{33}}\left(7\cdot\ket{\mathrm{f}_3 \conj{\mathrm{f}_3} \mathrm{f}_{\text{-}3}}-3\cdot\ket{\mathrm{f}_3 \mathrm{f}_2 \conj{\mathrm{f}_{\text{-}2}}}-2\cdot\ket{\mathrm{f}_3 \conj{\mathrm{f}_2} \mathrm{f}_{\text{-}2}}+3\cdot\ket{\mathrm{f}_3 \mathrm{f}_1 \conj{\mathrm{f}_{\text{-}1}}}\right.$ $\vphantom{\biggl(}$\\
&&&&$\left.-\ket{\mathrm{f}_3 \conj{\mathrm{f}_1} \mathrm{f}_{\text{-}1}}-2\cdot\ket{\mathrm{f}_3 \mathrm{f}_0 \conj{\mathrm{f}_0}}+5\cdot\ket{\conj{\mathrm{f}_3} \mathrm{f}_2 \mathrm{f}_{\text{-}2}}-2\cdot\ket{\conj{\mathrm{f}_3} \mathrm{f}_1 \mathrm{f}_{\text{-}1}}\right.$ $\vphantom{\biggl(}$\\
&&&&$\left.+\sqrt{15}\cdot\ket{\mathrm{f}_2 \conj{\mathrm{f}_2} \mathrm{f}_{\text{-}1}}-\sqrt{2}\cdot\ket{\mathrm{f}_2 \mathrm{f}_1 \conj{\mathrm{f}_0}}-\sqrt{2}\cdot\ket{\mathrm{f}_2 \conj{\mathrm{f}_1} \mathrm{f}_0}+2 \sqrt{2}\cdot\ket{\conj{\mathrm{f}_2} \mathrm{f}_1 \mathrm{f}_0}\right)$ $\vphantom{\biggl(}$\\
\cline{2-5}
&$^4\mathrm{D}^{\mathrm{o}}$&$2$&$\frac{3}{2}$&$\frac{1}{\sqrt{21}}\left(\sqrt{10}\cdot\ket{\mathrm{f}_3 \mathrm{f}_2 \mathrm{f}_{\text{-}3}}-\sqrt{6}\cdot\ket{\mathrm{f}_3 \mathrm{f}_1 \mathrm{f}_{\text{-}2}}+\sqrt{5}\cdot\ket{\mathrm{f}_3 \mathrm{f}_0 \mathrm{f}_{\text{-}1}}\right)$ $\vphantom{\biggl(}$\\
\cline{2-5}
&$^2\mathrm{D}^{\mathrm{o}}$&$2$&$\frac{1}{2}$&$\frac{1}{2 \sqrt{42}}\left(2 \sqrt{5}\cdot\ket{\mathrm{f}_3 \mathrm{f}_2 \conj{\mathrm{f}_{\text{-}3}}}-\sqrt{5}\cdot\ket{\mathrm{f}_3 \conj{\mathrm{f}_2} \mathrm{f}_{\text{-}3}}-2 \sqrt{3}\cdot\ket{\mathrm{f}_3 \mathrm{f}_1 \conj{\mathrm{f}_{\text{-}2}}}-\sqrt{3}\cdot\ket{\mathrm{f}_3 \conj{\mathrm{f}_1} \mathrm{f}_{\text{-}2}}\right.$ $\vphantom{\biggl(}$\\
&&&&$\left.+\sqrt{10}\cdot\ket{\mathrm{f}_3 \mathrm{f}_0 \conj{\mathrm{f}_{\text{-}1}}}+\sqrt{10}\cdot\ket{\mathrm{f}_3 \conj{\mathrm{f}_0} \mathrm{f}_{\text{-}1}}-\sqrt{5}\cdot\ket{\conj{\mathrm{f}_3} \mathrm{f}_2 \mathrm{f}_{\text{-}3}}+3 \sqrt{3}\cdot\ket{\conj{\mathrm{f}_3} \mathrm{f}_1 \mathrm{f}_{\text{-}2}}\right.$ $\vphantom{\biggl(}$\\
&&&&$\left.-2 \sqrt{10}\cdot\ket{\conj{\mathrm{f}_3} \mathrm{f}_0 \mathrm{f}_{\text{-}1}}+2 \sqrt{5}\cdot\ket{\mathrm{f}_2 \conj{\mathrm{f}_2} \mathrm{f}_{\text{-}2}}-\sqrt{5}\cdot\ket{\mathrm{f}_2 \conj{\mathrm{f}_1} \mathrm{f}_{\text{-}1}}+\sqrt{5}\cdot\ket{\conj{\mathrm{f}_2} \mathrm{f}_1 \mathrm{f}_{\text{-}1}}\right.$ $\vphantom{\biggl(}$\\
&&&&$\left.+\sqrt{6}\cdot\ket{\mathrm{f}_1 \conj{\mathrm{f}_1} \mathrm{f}_0}\right)$ $\vphantom{\biggl(}$\\
\cline{2-5}
&$^2\mathrm{D}^{\mathrm{o}}$&$2$&$\frac{1}{2}$&$\frac{1}{6 \sqrt{154}}\left(-14 \sqrt{5}\cdot\ket{\mathrm{f}_3 \mathrm{f}_2 \conj{\mathrm{f}_{\text{-}3}}}+7 \sqrt{5}\cdot\ket{\mathrm{f}_3 \conj{\mathrm{f}_2} \mathrm{f}_{\text{-}3}}+14 \sqrt{3}\cdot\ket{\mathrm{f}_3 \mathrm{f}_1 \conj{\mathrm{f}_{\text{-}2}}}-13 \sqrt{3}\cdot\ket{\mathrm{f}_3 \conj{\mathrm{f}_1} \mathrm{f}_{\text{-}2}}\right.$ $\vphantom{\biggl(}$\\
&&&&$\left.-\sqrt{10}\cdot\ket{\mathrm{f}_3 \mathrm{f}_0 \conj{\mathrm{f}_{\text{-}1}}}+5 \sqrt{10}\cdot\ket{\mathrm{f}_3 \conj{\mathrm{f}_0} \mathrm{f}_{\text{-}1}}+7 \sqrt{5}\cdot\ket{\conj{\mathrm{f}_3} \mathrm{f}_2 \mathrm{f}_{\text{-}3}}-\sqrt{3}\cdot\ket{\conj{\mathrm{f}_3} \mathrm{f}_1 \mathrm{f}_{\text{-}2}}\right.$ $\vphantom{\biggl(}$\\
&&&&$\left.-4 \sqrt{10}\cdot\ket{\conj{\mathrm{f}_3} \mathrm{f}_0 \mathrm{f}_{\text{-}1}}+6 \sqrt{5}\cdot\ket{\mathrm{f}_2 \conj{\mathrm{f}_2} \mathrm{f}_{\text{-}2}}-12 \sqrt{5}\cdot\ket{\mathrm{f}_2 \mathrm{f}_1 \conj{\mathrm{f}_{\text{-}1}}}+3 \sqrt{5}\cdot\ket{\mathrm{f}_2 \conj{\mathrm{f}_1} \mathrm{f}_{\text{-}1}}\right.$ $\vphantom{\biggl(}$\\
&&&&$\left.+12 \sqrt{5}\cdot\ket{\mathrm{f}_2 \mathrm{f}_0 \conj{\mathrm{f}_0}}+9 \sqrt{5}\cdot\ket{\conj{\mathrm{f}_2} \mathrm{f}_1 \mathrm{f}_{\text{-}1}}+9 \sqrt{6}\cdot\ket{\mathrm{f}_1 \conj{\mathrm{f}_1} \mathrm{f}_0}\right)$ $\vphantom{\biggl(}$\\
\cline{2-5}
&$^2\mathrm{P}^{\mathrm{o}}$&$1$&$\frac{1}{2}$&$\frac{1}{2 \sqrt{21}}\left(\sqrt{6}\cdot\ket{\mathrm{f}_3 \conj{\mathrm{f}_1} \mathrm{f}_{\text{-}3}}+\sqrt{3}\cdot\ket{\mathrm{f}_3 \mathrm{f}_0 \conj{\mathrm{f}_{\text{-}2}}}-2 \sqrt{3}\cdot\ket{\mathrm{f}_3 \conj{\mathrm{f}_0} \mathrm{f}_{\text{-}2}}-\sqrt{10}\cdot\ket{\mathrm{f}_3 \mathrm{f}_{\text{-}1} \conj{\mathrm{f}_{\text{-}1}}}\right.$ $\vphantom{\biggl(}$\\
&&&&$\left.-\sqrt{6}\cdot\ket{\conj{\mathrm{f}_3} \mathrm{f}_1 \mathrm{f}_{\text{-}3}}+\sqrt{3}\cdot\ket{\conj{\mathrm{f}_3} \mathrm{f}_0 \mathrm{f}_{\text{-}2}}-\sqrt{10}\cdot\ket{\mathrm{f}_2 \conj{\mathrm{f}_2} \mathrm{f}_{\text{-}3}}-\sqrt{6}\cdot\ket{\mathrm{f}_2 \mathrm{f}_1 \conj{\mathrm{f}_{\text{-}2}}}\right.$ $\vphantom{\biggl(}$\\
&&&&$\left.+\sqrt{6}\cdot\ket{\mathrm{f}_2 \conj{\mathrm{f}_1} \mathrm{f}_{\text{-}2}}+\sqrt{5}\cdot\ket{\mathrm{f}_2 \mathrm{f}_0 \conj{\mathrm{f}_{\text{-}1}}}-\sqrt{5}\cdot\ket{\conj{\mathrm{f}_2} \mathrm{f}_0 \mathrm{f}_{\text{-}1}}-\sqrt{6}\cdot\ket{\mathrm{f}_1 \conj{\mathrm{f}_1} \mathrm{f}_{\text{-}1}}\right.$ $\vphantom{\biggl(}$\\
&&&&$\left.-\sqrt{6}\cdot\ket{\mathrm{f}_1 \mathrm{f}_0 \conj{\mathrm{f}_0}}\right)$ $\vphantom{\biggl(}$\\
\cline{2-5}
&$^4\mathrm{S}^{\mathrm{o}}$&$0$&$\frac{3}{2}$&$\frac{1}{\sqrt{7}}\left(-\ket{\mathrm{f}_3 \mathrm{f}_0 \mathrm{f}_{\text{-}3}}+\sqrt{2}\cdot\ket{\mathrm{f}_3 \mathrm{f}_{\text{-}1} \mathrm{f}_{\text{-}2}}+\sqrt{2}\cdot\ket{\mathrm{f}_2 \mathrm{f}_1 \mathrm{f}_{\text{-}3}}-\ket{\mathrm{f}_2 \mathrm{f}_0 \mathrm{f}_{\text{-}2}}+\ket{\mathrm{f}_1 \mathrm{f}_0 \mathrm{f}_{\text{-}1}}\right)$ $\vphantom{\biggl(}$\\
\hline
\end{tabular}
\caption{Irreducible LS eigenspace decompositions of $\wedge^n V_{\mathrm{f}}$ for $n = 1,2,3$, see equation~\eqref{eq:IrredDecompose}. For conciseness, the table shows states with maximal $L_z$ and $\mathrm{S}_z$ quantum numbers only.}
\label{tab:irredLSf}
\end{table}

\section{Complexity analysis}
\label{sec:Costs}

This section contains a derivation of Eqs.~\eqref{eq:RunTime} and~\eqref{eq:d_un} in the limit of fixed electron number $n$ and $u \to \infty$.

\begin{figure}[!ht]
\centering
\includegraphics[width=0.8\textwidth]{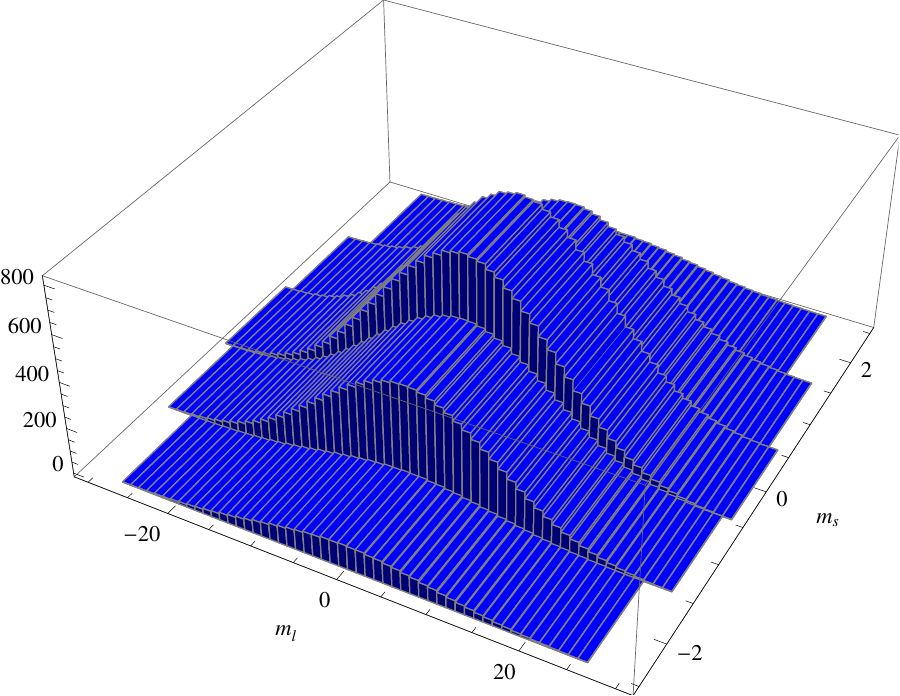}
\caption{Histogram plot of the $L_z$-$S_z$ eigenvalue multiplicities of $\wedge^4 V_8$. This is equivalent to the table in figure~\ref{fig:quantLS} but for $u = 8$ and $n = 4$. The probability density function approaches a normal distribution as a result of the central limit theorem for large $n$; compare with proposition~\ref{prop:LzSzDist}.}
\label{fig:LzSzDist}
\end{figure}

We first investigate the multiplicity distribution of the simultaneous $L_z$-$S_z$ eigenvalues, as illustrated in figure~\ref{fig:LzSzDist}. In the following, $T(m_\ell,m_s)$ denotes the multiplicity of the simultaneous $L_z$-$S_z$ eigenspace with eigenvalues $(m_\ell,m_s)$ on $\wedge^n V_u$. We write $\round{\cdot}$ for the nearest integer function. Furthermore, $f_{\mathrm{IH},n}$ and $f_{\mathrm{bin},n,p}$ denote the probability density functions of the standard Irwin--Hall distribution~\cite{Irwin1927,Hall1927} (sum of $n$ i.i.d.\ $U(0,1)$ random variables) and the binomial distribution with parameters $(n,p)$, respectively.
\begin{proposition}
\label{prop:LzSzDist}
Given a fixed integer $n \ge 1$, define
\begin{equation}
t_{u,n}(x_\ell,m_s) := \frac{u\,T(\round{u\,x_\ell},m_s)}{\dim\left(\wedge^n V_u\right)}, \quad x_\ell \in [-n,n], \quad m_s \in \left\{n/2,\dots,-n/2\right\}.
\end{equation}
Then for each $m_s$,
\begin{equation}
\label{eq:MultiplicityDist}
\lim_{u \to \infty} t_{u,n}(x_\ell,m_s) = f_{L_z}\!(x_\ell) f_{S_z}\!(m_s)
\end{equation}
uniformly in $x_\ell$ with
\begin{equation}
\label{eq:fLzfSz}
f_{L_z}\!(x_\ell) := \frac{1}{2}\,f_{\mathrm{IH},n}\!\left(\frac{x_\ell}{2} + \frac{n}{2}\right), \quad f_{S_z}\!(m_s) := f_{\mathrm{bin},n,\frac{1}{2}}\left(m_s+\frac{n}{2}\right).
\end{equation}
In particular, $f_{L_z}$ and $f_{S_z}$ have zero mean and variances $\sigma_\ell^2 = n/3$ and $\sigma_s^2 = n/4$, respectively.
\end{proposition}

The factor $u$ in the definition of $t_{u,n}$ ensures normalization in the sense that
\begin{equation}
\begin{split}
&\sum_{m_s} \int_{[-n,n]} t_{u,n}(x_\ell,m_s) \ud x_\ell
= \dim\left(\wedge^n V_u\right)^{-1} \sum_{m_s} \int_{[-n\,u,n\,u]} T(\round{m_\ell},m_s) \ud m_\ell\\
&\qquad\cong \dim\left(\wedge^n V_u\right)^{-1} \sum_{m_\ell,m_s} T(m_\ell,m_s) = 1.
\end{split}
\end{equation}

\begin{proof}
First label the basis vectors (``spherical harmonics'') spanning $V_u$ abstractly as 
\begin{equation}
Y_u := \left\{u\!\uparrow,u\!\downarrow,\dots,(-u)\!\uparrow,(-u)\!\downarrow\right\}.
\end{equation}
Now let $\psi = \ket{\varphi_1,\dots,\varphi_n} \in \wedge^n V_u$ be a uniformly random Slater determinant, with $\varphi_i \in Y_u$ pairwise different. In other words, $\psi$ randomly selects $n$ distinct elements from $Y_u$. As already shown in the beginning of section~\ref{sec:Algorithm}, $\psi$ is a simultaneous $L_z$-$S_z$ eigenvector. To estimate the distribution $(\mathrm{eig}_{L_z}\!(\psi),\mathrm{eig}_{S_z}\!(\psi))$, note that $L_z$ and $S_z$ just sum up the corresponding terms in $\psi$. Thus, for example,
\begin{align}
\mathrm{eig}_{L_z}\!(\ket{2\!\uparrow,1\!\downarrow,-1\!\uparrow}) &= 2 + 1 - 1 = 2,\\
\mathrm{eig}_{S_z}\!(\ket{2\!\uparrow,1\!\downarrow,-1\!\uparrow}) &= \tfrac{1}{2} - \tfrac{1}{2} + \tfrac{1}{2} = \tfrac{1}{2}.
\end{align}
Observe that the error incurred by ignoring the exclusion principle goes to zero as $u \to \infty$ due to $n \ll u$. That is, we may replace $\psi$ by $\tilde{\psi} := \tilde{\varphi}_1 \otimes \cdots \otimes\tilde{\varphi}_n \in \bigotimes^n V_u$ with $\tilde{\varphi}_i \in Y_u$ i.i.d.\ (independent and identically distributed). Then $\mathrm{eig}_{L_z}\!(\tilde{\psi})$ and $\mathrm{eig}_{S_z}\!(\tilde{\psi})$ are independent as well and can be handled separately. The distribution $f_{S_z}$ stems directly from $\mathrm{eig}_{S_z}\!(\tilde{\psi}) = \sum_i \mathrm{eig}_{S_z}\!(\tilde{\varphi}_i)$. Considering $\mathrm{eig}_{L_z}\!(\tilde{\psi})$, first note that the discretization error
\begin{equation}
\left\lvert f_{L_z}\!(x_\ell) - f_{L_z}\!\left(\frac{\round{u\,x_\ell}}{u}\right) \right\rvert \to 0
\end{equation}
as $u \to \infty$ since $f_{L_z}$ is uniformly continuous. Thus, the distribution of $\frac{1}{u}\,\mathrm{eig}_{L_z}\!(\tilde{\varphi}_i)$ approaches $U(-1,1)$, and consequently, $\frac{1}{u}\,\mathrm{eig}_{L_z}\!(\tilde{\psi}) = \sum_i \frac{1}{u}\,\mathrm{eig}_{L_z}\!(\tilde{\varphi}_i) \sim f_{L_z}$.
\end{proof}

\subsection{Run time}
\label{sec:RunTime}

This subsection is concerned with the asymptotic run time of the main algorithm, as already stated in the introduction.

\begin{proposition}
For any fixed integer $n \ge 1$, the run time $R_n(u)$ of algorithm~\ref{alg:IrredLS} obeys
\begin{equation}
\label{eq:RunTime}
R_n(u) = \mathcal{O}\big(u^{3n-2}\big)
\end{equation}
as $u \to \infty$.
\end{proposition}
\begin{proof}
Due to the sparse matrix structure of the lowering operators $L_-$ and $S_-$, each matrix multiplication in line~\ref{algline:LoweringSpan} of the algorithm has linear (instead of quadratic) cost. Thus, the main computational cost stems from line~\ref{algline:OrthComplement}. Denote the tuples $(\ell_i,s_i)$ after deleting duplicates by $(\ell'_k,s'_k)$. For each simultaneous $L_z$-$S_z$ eigenspace $W_{\ell'_k,s'_k}$ with dimension $d_k := \dim(W_{\ell'_k,s'_k})$, the algorithm calculates $d_k$ orthogonal complements within $W_{\ell'_k,s'_k}$, each of which takes $\mathcal{O}(d_k^2)$ operations. So in total, $R_n(u) = \mathcal{O}(\sum_k d_k^3)$. Combining this result with~\eqref{eq:MultiplicityDist} yields the following upper bound,
\begin{equation}
\label{eq:RunTimeU}
\begin{split}
R_n(u)
&\lesssim \frac14 \dim\left(\wedge^n V_u\right)^3 \int_{[-n\,u,n\,u]} u^{-3}\, f_{L_z}(m_\ell/u)^3 \ud m_\ell \,\sum_{m_s} f_{S_z}(m_s)^3\\
&= \frac14 \dim\left(\wedge^n V_u\right)^3 u^{-2} \int_{\mathbb{R}} f_{\mathrm{IH},n}(x_\ell)^3 \ud x_\ell \sum_{x_s = 0}^n f_{\mathrm{bin},n,\frac{1}{2}}(x_s)^3\\
&= \mathcal{O}\big(u^{3n-2}\big).
\end{split}
\end{equation}
The factor $\frac14$ stems from the observation that for each $k$, neither $W_{\ell'_k,-s'_k}$, $W_{-\ell'_k,s'_k}$ nor $W_{-\ell'_k,-s'_k}$ contribute to the cost. The second line follows from a change of variables, and the third from noting that the integral and sum in the second line do not depend on $u$.
\end{proof}

Taking one step further, we can now investigate the dependency of $R_n(u)$ on $n$ in more detail and evaluate the terms in the second line of~\eqref{eq:RunTimeU}. We obtain the following
\begin{lemma}
Assume that $n$ is large enough such that $f_{L_z}$ and $f_{S_z}$ can be well approximated by Gaussian normal distributions with mean $0$ and variances $\sigma_\ell$ and $\sigma_s$ from proposition~\ref{prop:LzSzDist}. Then
\begin{equation}
R_n(u) \lesssim \frac{\dim\left(\wedge^n V_u\right)^3}{48 \pi^2\,u^2\,\sigma_\ell^2\,\sigma_s^2} = \frac{\dim\left(\wedge^n V_u\right)^3}{(2 \pi\,n\,u)^2}.
\end{equation}
\end{lemma}

\subsection{Dimension of the central $L_z$-$S_z$ eigenspace}
\label{sec:d_un}

Let $d_{u,n}$ label the maximum dimension of any simultaneous $L_z$-$S_z$ eigenspace on $\wedge^n V_u$, which is attained by the ``central'' eigenspace with eigenvalues $(m_\ell,m_s) = (0,0)$ for $n$ even and $(0,\frac{1}{2})$ for $n$ odd, respectively. Thus, $d_{u,n}$ can be approximated by evaluating the right side of equation~\eqref{eq:MultiplicityDist} at these eigenvalues. A comparison with the exact $d_{u,n}$ is shown in figure~\ref{fig:d_un}, which nicely illustrates the polynomial scaling in $u$. As a remark, $f_{\mathrm{IH},n}(\frac{n}{2}) = \frac{2}{\pi} \int_0^\infty \mathrm{sinc}(x)^n \ud x$ due to the convolution theorem applied to the uniform probability density function on the interval $[-1/2,1/2]$.

\begin{figure}[!ht]
\centering
\includegraphics[width=0.8\textwidth]{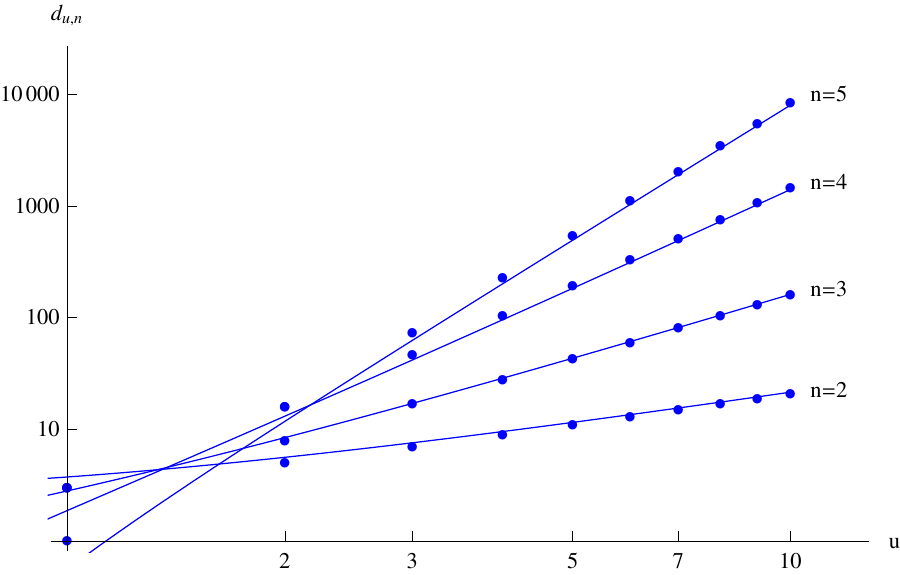}
\caption{Log-log plot of $d_{u,n}$ versus $u$ for various $n$. Dots are exact values, and lines show the right side of equation~\eqref{eq:MultiplicityDist} evaluated at $(m_\ell,m_s) = (0,0)$ for $n$ even and $(0,\frac{1}{2})$ for $n$ odd, respectively.}
\label{fig:d_un}
\end{figure}

To derive equation~\eqref{eq:d_un}, we follow the same procedure as above and replace $f_{L_z}$ and $f_{S_z}$ by Gaussian normal distributions. We then set $m_s = 0$ both for $n$ even and $n$ odd since $\frac{1}{2}$ is small compared to $n$. Plugging in $(m_\ell,m_s) = (0,0)$ yields
\begin{lemma}
Assume that $n$ is large enough such that $f_{L_z}$ and $f_{S_z}$ can well be approximated by Gaussian normal distributions. Then
\begin{equation}
d_{u,n} \cong \frac{\dim\left(\wedge^n V_u\right)}{2\pi\,u\,\sigma_\ell\,\sigma_s} = \sqrt3\ \frac{\dim\left(\wedge^n V_u\right)}{\pi\,n\,u}.
\end{equation}
\end{lemma}

\section{Conclusions}
\label{sec:Conclusions}

The main principle of the algorithm is the \emph{implicit} simultaneous diagonalization of the many-particle angular momentum, spin and parity operators by algebraic traversal of the $L_z$-$S_z$ eigenstates in the correct order. This involves $\mathcal{O}(u^{3n-2})$ operations for angular subshell $u$ filled with $n$ electrons. When taking any admissible $n$ into account, subshells up to $u = \mathrm{d}$ are feasible at present, and $u = \mathrm{f}$ for all $n = 1, \dots, 14$ might still be attainable. Notably, the electronic ground state configurations found in the periodic table are precisely constructed from the atomic $\mathrm{s}$, $\mathrm{p}$, $\mathrm{d}$, $\mathrm{f}$ subshells.

\paragraph{Acknowledgments.} I would like to thank Gero Friesecke for many helpul discussions, and DFG for financial support under project FR 1275/3-1.


\end{document}